%% file: main_vertex_conn.tex
\documentclass[11pt]{article}

\input{header.tex}

\title{Deterministic $k$-Vertex Connectivity in $k^2$ Max-flows}

\author{
           Chaitanya Nalam\thanks{\texttt{nalamsai@umich.edu}. University of Michigan, USA}
           \and
         Thatchaphol
         Saranurak\thanks{\texttt{thsa@umich.edu}. University of Michigan, USA. Supported by NSF CAREER grant 2238138.}
         \and
         Sorrachai Yingchareonthawornchai\thanks{\texttt{sorrachai.yingchareonthawornchai@aalto.fi}. Aalto University, Finland}
}

\date{}   
\begin{document}

\maketitle
\pagenumbering{gobble}
\input{abstract_thatchaphol}
\newpage        

\pagenumbering{arabic}

\input{intro}

\input{prelims}
\input{exact}
\input{approxImprov}
\input{open}

\section*{Acknowledgement} 
This project has received funding from the European Research Council (ERC) under the European Union's Horizon 2020 research and innovation programme under grant agreement No 759557. 

\small{
\bibliographystyle{alpha}
\bibliography{refs,references} 
}

\end{document}

%% file: header.tex
\usepackage{fullpage}
\usepackage[utf8]{inputenc} 
\usepackage[T1]{fontenc}

\usepackage{geometry}
\usepackage{physics}

\usepackage{bbm}
\usepackage{algorithmic} 
\usepackage{microtype}  
\usepackage{array}
\usepackage{multirow}
\usepackage{amsmath} 
\usepackage{amssymb}
\usepackage{amsthm}
\usepackage{thmtools}
\usepackage{thm-restate}
\usepackage{mathtools}
\DeclarePairedDelimiter\ceil{\lceil}{\rceil}
\DeclarePairedDelimiter\floor{\lfloor}{\rfloor}
\usepackage{todonotes}
\setuptodonotes{inline}

\usepackage[procnumbered,ruled,vlined,linesnumbered]{algorithm2e}
 
\SetCommentSty{mycommfont} %

\usepackage{xcolor}  
\usepackage{xspace}
\usepackage{enumitem}    
\usepackage{tikz}   

\usepackage{comment} 

\usepackage{graphicx}
\graphicspath{{Figures/}}
\usepackage{caption}

\usepackage{subcaption}

\usepackage{xcolor}
\usepackage{nameref}
\definecolor{ForestGreen}{rgb}{0.1333,0.5451,0.1333}
\definecolor{DarkRed}{rgb}{0.65,0,0}
\definecolor{Red}{rgb}{1,0,0}
\usepackage[linktocpage=true,
pagebackref=true,colorlinks,
linkcolor=DarkRed,citecolor=ForestGreen,
bookmarks,bookmarksopen,bookmarksnumbered]
{hyperref}
\usepackage[noabbrev,nameinlink]{cleveref}

\AtBeginDocument{%
  \DeclareFontShape{T1}{cmr}{m}{scit}{<->ssub*cmr/m/sc}{}%
  \DeclareFontShape{T1}{cmr}{bx}{sc}{<->ssub*cmr/m/n}{}
}

\usepackage[procnumbered,ruled,vlined,linesnumbered]{algorithm2e}
\SetCommentSty{mycommfont} %
\SetKwInput{KwGlobalVar}{Global variables}

\declaretheorem[numberwithin=section]{theorem}
\declaretheorem[numberlike=theorem]{lemma}

\declaretheorem[numberlike=theorem]{fact}

\declaretheorem[numberlike=theorem]{claim}

\crefname{algorithm}{Algorithm}{Algorithms}
\Crefname{algorithm}{Algorithm}{Algorithms}

    \newtheoremstyle{TheoremNum}
        {\topsep}{\topsep}              %
        {\itshape}                      %
        {}                              %
        {\bfseries}                     %
        {.}                             %
        { }                             %
        {\thmname{#1}\thmnote{ \bfseries #3}}%
    \theoremstyle{TheoremNum}

\theoremstyle{definition}
\declaretheorem[numberlike=theorem]{definition}

\newcommand{\ot}{\tilde{O}}

\newcommand{\phibar}{\bar\phi}

\newcommand{\eps}{\epsilon}

\SetKwFor{RepTimes}{repeat for}{times}{end}

\newcommand{\eat}[1]{}

\renewcommand{\mod}{\operatorname{mod}}

\newcommand{\cGabow}{c_{{\ref{lem:gabow00}}}}

\global\long\def\Ohat{\widehat{O}}

\def\ShowComment{True} %

\ifdefined\ShowComment
\def\thatchapholtext#1{\textcolor{purple}{#1}}
\def\thatchaphol#1{\marginpar{$\leftarrow$\fbox{T}}\footnote{$\Rightarrow$~{\sf\textcolor{purple}{#1 --Thatchaphol}}}}
\def\danupon#1{\textcolor{orange}{DN: #1}}
\def\sorrachai#1{\marginpar{$\leftarrow$\fbox{S}}\footnote{$\Rightarrow$~{\sf\textcolor{blue}{#1
-- Sorrachai}}}}

\def\note#1{#1}
\def\alert#1{\textcolor{red}{#1}}

\else
\def\thatchapholtext#1{}
\def\thatchaphol#1{}
\def\danupon#1{}
\def\shen#1{}
\def\sorrachai#1{}
\def\jason#1{}
\def\note#1{} 
\def\alert#1{}

\fi

%% file: abstract_thatchaphol.tex
\begin{abstract}
    
An $n$-vertex $m$-edge graph is \emph{$k$-vertex connected} if it cannot be disconnected by deleting less than $k$ vertices. After more than half a century of intensive research, the result by [Li~et~al.~STOC'21] finally gave a \emph{randomized} algorithm for checking $k$-connectivity in near-optimal $\widehat{O}(m)$ time.\footnote{We use $\widehat{O}(\cdot)$ to hide an $n^{o(1)}$ factor.} Deterministic algorithms, unfortunately, have remained much slower even if we assume a linear-time max-flow algorithm: they either require at least $\Omega(mn)$ time [Even'75; Henzinger Rao and Gabow, FOCS'96; Gabow, FOCS'00] or assume that $k=o(\sqrt{\log n})$ [Saranurak and Yingchareonthawornchai, FOCS'22].

We show a \emph{deterministic} algorithm for checking $k$-vertex connectivity in time proportional to making $\widehat{O}(k^{2})$ max-flow calls, and, hence, in $\widehat{O}(mk^{2})$ time using the deterministic max-flow algorithm by [Brand~et~al.~FOCS'23]. Our algorithm gives the first almost-linear-time bound for all $k$ where $\sqrt{\log n}\le k\le n^{o(1)}$ and subsumes up to a sub polynomial factor the long-standing state-of-the-art algorithm by [Even'75] which requires $O(n+k^{2})$ max-flow calls. 

Our key technique is a deterministic algorithm for terminal reduction for vertex connectivity: given a terminal set separated by a vertex mincut, output either a vertex mincut or a smaller terminal set that remains separated by a vertex mincut.

We also show a deterministic $(1+\epsilon)$-approximation algorithm for vertex connectivity that makes $O(n/\epsilon^2)$ max-flow calls, improving the bound of $O(n^{1.5})$ max-flow calls in the exact algorithm of [Gabow, FOCS'00]. The technique is based on Ramanujan graphs.

\end{abstract}

%% file: intro.tex
\section{Introduction}\label{sec:intro}

\textit{Vertex connectivity} of an $n$-vertex $m$-edge undirected graph $G$, denoted by $\kappa_G$, is the minimum number of vertices to be removed to disconnect the graph (or to become a singleton). Vertex connectivity along with its related problem called edge connectivity (where we remove edges to disconnect the graph) are both fundamental and very well-studied problems in graph algorithm research ~\cite{Menger1927,Kleitman1969methods,Podderyugin1973algorithm,EvenT75,Even75,Galil80,EsfahanianH84,Matula87,BeckerDDHKKMNRW82,LinialLW88,CheriyanT91,NagamochiI92,CheriyanR94,Henzinger97,HenzingerRG00,Gabow06,Censor-HillelGK14}, see \cite{NanongkaiSY19} for more discussion.  The complexity of the edge connectivity problem is well-understood: It can be solved in nearly linear time using randomization by Karger since 2000 \cite{Karger00} and, more recently, it can also be solved deterministically in near-linear time~\cite{KawarabayashiT15,HenzingerRW17} and even in weighted case \cite{LiP20deterministic,Li21mincut}. 

For vertex connectivity, there is still a large gap in our understanding.  Aho, Hopcroft, and Ullman asked almost 50 years ago if vertex connectivity can be solved in linear time~\cite{AhoHU74}. The answer is affirmative (up to a subpolynomial factor) using randomization: Recent developments~\cite{NanongkaiSY19,ForsterNYSY20,li_vertex_2021} culminate in a celebrated ``polylog max-flow'' time, which is  almost linear by the recent breakthrough in max-flow problem~\cite{ChenKLPGS22}. 

Deterministic algorithms, on the other hand, have remained much slower.
In the discussion below, we will assume a linear-time deterministic max-flow algorithm so that we can highlight the development of ideas related to vertex connectivity itself and not about the implementation of max-flow algorithms. This is actually without loss of generality because Brand~et~al.~\cite{BrandCKLPGSS23} recently showed an almost-linear-time deterministic max-flow algorithm.
Let us consider the problem of \textit{checking $k$-vertex connectivity} (i.e., deciding if vertex connectivity is at least $k$ or outputs a minimum vertex cut).\footnote{Note that, for all results stated below including our results, we can assume that $m = O(nk)$ by using the linear-time sparsification algorithm by \cite{NagamochiI92}.}
When $k \le 3$, $O(m)$-time algorithms are known  using a depth-first search~\cite{Tarjan72} and a  data structure called SPQR tree~\cite{HopcroftT73}. Very recently, \cite{SaranurakY22} showed an $\Ohat(m2^{O(k^2)})$-time algorithm which is almost-linear for all $k = o(\sqrt{\log n})$.  For general $k$, half a century ago, Kleitman \cite{Kleitman1969methods} showed an algorithm that makes $O(nk)$ max-flow calls. Then, Even \cite{Even75} improved the number of max-flow calls to $O(n+k^2)$. 
\cite{HenzingerRG00} later showed how to implement the first $n$ calls to max-flow in Even's algorithm in $O(mn)$ time, without assuming a linear-time max-flow algorithm. 
The currently fastest deterministic algorithm is by Gabow \cite{Gabow06}. The bound stated in his paper is $O(\min\{n^{3/4}, k^{1.5}\}km + mn)$ time, but this bound is based on max-flow algorithms that are slower than linear time. One can show that his algorithms take time proportional to $O(n + k \min\{k,\sqrt{n}\})$ max flow calls, which improved upon Even's algorithm when  $k \ge \sqrt{n}$. 

To summarize state of the art, for any $k \ge \sqrt{\log n}$, all known deterministic algorithms require $\Omega(n)$ max-flow calls, which is $\Omega(mn)$ time even if we assume a linear time max-flow algorithm.

\subsection{Our Results} 
In this paper, we show a \textit{deterministic} algorithm for checking $k$-vertex connectivity in $\Ohat(k^2)$ max-flows.
Let $G = (V,E)$ be an undirected graph. A \textit{vertex cut} $(L,S,R)$ is a partition of the vertex set $V$ such that there is no edge between $L$ and $R$. The \emph{size} of $(L,S,R)$ is $|S|$.

\begin{theorem} \label{thm:main}
There is a deterministic algorithm that takes as inputs an $n$-vertex $m$-edge undirected graph $G = (V,E)$, and connectivity parameter $k$, and either outputs a vertex cut of size less than $k$ or concludes that $G$ is $k$-vertex connected. The algorithm makes calls to unit-vertex-capacity max-flow instances with $\Ohat(mk^2)$ number of edges in total and spends an additional $\Ohat(mk^2)$ time.
\end{theorem}

In particular, our algorithm runs in $o(n)$ max-flows whenever $k \ll \sqrt{n}$. This is the first algorithm that improves the long-standing state-of-the-art algorithm~\cite{Even75} that runs in $O(n+k^2)$ max-flows. Using the almost linear time deterministic max-flow algorithm~\cite{BrandCKLPGSS23}, our algorithm runs in $\Ohat(mk^2)$ time, which is the fastest whenever $k \ll \sqrt{n}$. Our algorithm exponentially improves the dependency on $k$ from the recent $\Ohat(m2^{O(k^2)})$-time algorithm~\cite{SaranurakY22}. %

We also show a $(1+\epsilon)$-approximation algorithm for computing vertex connectivity of $G$, $\kappa_G$.

\begin{restatable}[]{theorem}{approximationthm}
\label{thm:approx}
There is a deterministic algorithm that takes an $n$-vertex $m$-edge undirected graph $G$ and accuracy parameter $\epsilon \in (0,1)$ as inputs and outputs a vertex cut of size at most $\floor{(1+\epsilon)\kappa_G}$. The algorithm makes $O(n/\eps^2)$ calls to unit-vertex-capacity $(s,t)$-max-flows on $G$ and spends additional $O(n/\eps^2)$ time.
\end{restatable}

In particular, this algorithm runs in $\Ohat(\frac{1}{\epsilon^2}\cdot mn)$ time using the max-flow algorithm by \cite{BrandCKLPGSS23}.
This is the first deterministic approximation algorithm that runs in $\ll mn^{1.5}$ time, and improves upon the exact algorithm~\cite{Gabow06} that requires $n^{1.5}$ max-flows in the general case.

\subsection{Techniques}
The key to our result is the terminal reduction algorithm for vertex connectivity that runs in $\Ohat(k^2)$ max-flows. Let $G = (V,E)$ be an $n$-vertex $m$-edge undirected graph with terminal set  $T \subseteq V$.  A $T$-\textit{Steiner} vertex cut $(L,S,R)$ is a vertex cut such that $T\cap L \neq \emptyset$ and $T\cap R \neq \emptyset$. We denote $\kappa_G(T)$ as the size of the minimum $T$-Steiner vertex cut in $G$  or $n-1$ if it does not exist. By definition, vertex connectivity of $G$, denoted by $\kappa_G$ is equal to $\kappa_G(V)$. We omit the subscript when the context is clear.  Our main technical result is:

\begin{restatable}[Terminal reduction]{theorem}{terminalreduction}
\label{thm:main terminal reduction}
There is a deterministic algorithm that takes as inputs an $n$-vertex $m$-edge  undirected graph $G = (V,E)$, along with a terminal set $T \subseteq V$ and a cut parameter $k$, and outputs either a vertex cut of size less than $k$ or a new terminal set $T' \subseteq V$ of size at most $|T|/2$ such that $\kappa(T') < k$ if $\kappa(T) < k$.   The algorithm makes calls to unit-vertex-capacity max-flow instances with $\Ohat(mk^2)$ number of edges in total and spends an additional  $\Ohat(mk^2)$ time.
\end{restatable}

Intuitively, terminal reduction means we find a smaller terminal set $T'$ such that the minimum $T$-Steiner cut of size less than $k$ remains $T'$-Steiner. Given \Cref{thm:main terminal reduction},  our main theorem follows immediately.

\begin{proof}[Proof of \Cref{thm:main}]
To check $k$-vertex connectivity of a graph $G=(V,E)$, we start with the entire vertex set $V$ as a terminal set, and repeatedly apply \Cref{thm:main terminal reduction} for $O(\log n)$ times until we obtain a vertex cut of size less than $k$ or the terminal set becomes empty ($T = \emptyset$) for which we conclude that $G$ is $k$-vertex connected.
\end{proof}

Prior to our work, \cite{LiP20deterministic} (implicitly) showed exactly the same terminal reduction algorithm as in \Cref{thm:main terminal reduction} except their algorithm is for edge mincuts and they guaranteed that $T' \subseteq T$, i.e., the output terminal $T'$ is always a subset of the input terminal $T$. This is \emph{not} guaranteed by our \Cref{thm:main terminal reduction}.%

\paragraph{Terminal Reduction Framework.} The key method to the proof of \Cref{thm:main terminal reduction} is to generalize the vertex reduction framework in \cite{SaranurakY22} to the terminal setting. 
At a high-level, the vertex reduction framework in \cite{SaranurakY22} consists of two steps. %
\begin{enumerate}
    \item First, given a graph $G=(V,E)$, their algorithm 
    outputs either a vertex cut of size $<k$ or a terminal set $T \subseteq V$ of size  $\ll n$ such that $\kappa_G(T) < k$ iff  $\kappa_G < k$. 
    \item In the second step, which is a vertex reduction step, given the terminal set $T$ from the first step, their algorithm  outputs a smaller graph $H$ of size $\Ohat(|T|) < m/2$ such that $\kappa_G < k$ iff $\kappa_H < k$. 
\end{enumerate}  The two-step vertex reduction framework~\cite{SaranurakY22} allows them to solve $k$-vertex connectivity after repeatedly applying it for $O(\log n)$ times. The drawback in their algorithm is that their second step requires $\Ohat(m  2^{O(k^2)})$ time, which is very slow when $k \gg \sqrt{\log n}$. %

In this paper, we consider the generalization, of the first step of the vertex reduction framework~\cite{SaranurakY22}, to the terminal setting:
\begin{quote} 
Given a graph $G = (V,E)$ and its terminal set $T \subseteq V$, output either a vertex cut of size less than $k$ or a terminal set $T'$ of size at most $|T|/2$ such that $\kappa(T') < k $ iff $\kappa(T) < k$. 
\end{quote}
 As a result, the new algorithm bypasses the second step of the vertex reduction framework~\cite{SaranurakY22} by directly reducing the terminal set while maintaining the Steiner connectivity of the output terminal set. Thus, we avoid spending $\Ohat(m 2^{O(k^2)})$ time for vertex reduction as in \cite{SaranurakY22}. By \Cref{thm:main terminal reduction}, we only spend $\Ohat(m k^2)$ time instead (using a deterministic linear-time max-flow algorithm~\cite{BrandCKLPGSS23}). The bottleneck of $k^2$ factor stems from the key subroutine in our terminal reduction that handles \textit{terminal-unbalanced} vertex cuts which we discuss next. We say that a vertex cut $(L,S,R)$ is \textit{terminal-unbalanced} if $\min\{|T \cap L|, |T \cap R|\} \leq n^{o(1)}$. 

\paragraph{Handling Terminal-unbalanced Vertex Cuts.} For simplicity, consider the following task:
\begin{quote}
    Given a terminal set $T \subseteq V$ with a promised that there is a vertex cut $(L,S,R)$ of size less than $k$ such that $|T \cap L| = 1, |T \cap R| \geq 1$, design a fast algorithm that outputs a vertex cut of size less than $k$. 
\end{quote}

If the vertex cut $(L,S,R)$ with the above promise is such that $S \cap T = \emptyset$, then we say that $(L,S,R)$ \emph{isolates} $T$. In this case, we can apply the isolating vertex cut lemma~\cite{li_vertex_2021} to deterministically output a vertex cut of size less than $k$ in time proportional to $O(\log |T|)$ calls of max-flows and we are done. In general, the vertex cut $(L,S,R)$ may not be $T$-isolating because possibly $S\cap T \neq \emptyset$.  As a result, it remains to solve the following task.

\begin{quote}
    Find a subset $T' \subseteq T$ where $(L,S,R)$ isolates $T'$, i.e., $|T'\cap L| = 1, |T'\cap S| = 0, |T'\cap R| \geq 1$. 
\end{quote}

To do so, we construct a family $\mathcal{F}$ of subsets $U \subseteq T$ of size $|\mathcal{F}|\le k^2 \log^{O(1)} n$ such that $(L,S,R)$ must isolate some $U' \in \mathcal{F}$ via a  \textit{Hit and miss hash family}~\cite{s_deterministic_2023}. Given such a family $\mathcal{F}$, we apply the isolating vertex cut lemma on each $U \in \mathcal{F}$ to obtain a vertex cut of size less than $k$ as desired. The total running time is roughly  $k^2 \log^{O(1)}n\cdot O(\log |T|) = k^2 \log^{O(1)}n$ calls to max-flows.

\paragraph{Using Ramanujan graphs for Approximate Mincuts.}
The goal is to find a vertex cut of size at most $(1+\eps)\kappa_G$. 
If the minimum degree is at most $(1+\eps)\kappa_G$, we can simply return the neighbor set of the min-degree vertex, sacrificing a $(1+\eps)$ factor. Otherwise, it is easy to show that any minimum cut $(L, S, R)$ is such that $|L|,|R|\ge \eps \kappa_G$, i.e., the cut is quite balanced. We show that Ramanujan graphs can give us a list of $O(n/\eps^2)$ vertex pairs such that a pair $(s,t)$ must exist where $s \in L, t \in R$. Thus, we will identify a minimum vertex cut by running unit-vertex-capacity $(s,t)$-max-flows on graph $G$ for each of the $O(n/\eps^2)$ vertex pairs. 

To obtain the above list of vertex pairs, let $X = (V,E_X)$ be a regular Ramanujan graph with degree $\Theta(1/\eps^2)$.  Assuming $|L|\le |R|$, we want to show that there exists an edge $(s,t)\in E_X$ where $s\in L$ and $t\in R$. 
There are two cases. If $|L| \ge \Omega(\eps n)$, then such an edge $(s,t) \in E_X$ must exist by the \emph{expander mixing lemma} as observed before in \cite{Gabow06}. Otherwise, we use the known fact that strong spectral expansion of $X$ implies strong vertex expansion for all ``small'' sets of $X$. Since $|L| \le O(\eps n)$, one can show that the vertex expansion of $L$ is at least $2/\eps$, that is $|N_X(L)|\ge 2|L|/\eps \ge 2\kappa_G$ where $N_X(L)$ denotes the set of neighbors of $L$ in $X$. So $N_X(L) \setminus (L\cup S) \neq \emptyset$, i.e., $N_X(L) \cap R \neq \emptyset$. This again implies an edge $(s,t)\in E_X$ where $s\in L$ and $t\in R$.

\subsection{Organization}
We start with introducing the notations needed for the rest of the paper in \Cref{sec:prelims}. In \Cref{sec:exact}, we prove our main technical result which is the fast terminal reduction algorithm  (\Cref{thm:main terminal reduction}) that implies our main result \Cref{thm:main}.
In \Cref{sec:approx}, we use the Ramanujan expanders to obtain a faster approximation algorithm, proving \Cref{thm:approx}.

%% file: prelims.tex
\section{Preliminaries}\label{sec:prelims}
When we say a graph $G = (V,E)$, we denote $n = |V|$ and $m = |E|$ unless stated otherwise. Let $A,B\subseteq V$, we use $E(A,B)$ to denote the set of edges whose one endpoint belongs to $A$ and the other belongs to $B$. For the rest of the paper, unless stated otherwise, when we say an algorithm, we mean a deterministic algorithm. When we say $k$-connected, we mean $k$-vertex connected. When we say a cut, we mean vertex cut.

For any vertex set $S$, we denote $G - S$ as the graph $G$ after removing all vertices in $S$. A vertex cut $(L,S,R)$ is a partition of the vertex set $V$ such that there is no edge between $L$ and $R$. We say that a vertex set $S$ is a \textit{vertex separator} (or simply a separator) if $G - S$ is disconnected. By definition, the vertex set $S$ in a vertex cut $(L,S,R)$ is a separator. We say that $S$ \textit{separates} $x$ and $y$ for some $x,y \in V$ if $x,y \not \in S$ and there is no path from $x$ to $y$ in $G - S$. The minimum number of vertices we need to remove to separate some pair of vertices in the graph $G = (V,E)$ is called vertex connectivity denoted by $\kappa_G$ and those set of vertices $S$ is called vertex mincut.

Let $u,v\in V$, we use $\kappa(u,v)$ to denote the minimum number of vertices need to be removed to disconnect $u$ from $v$ (or $n-1$ if it is not possible). Let $T\subseteq V$ be a set of terminals, we use 
$\kappa(T) = \min_{x,y\in T}\kappa(x,y)$ to denote the \emph{steiner vertex connectivity} of terminal set $T$ which is the minimum number of vertices to be removed to disconnect at least one terminal from other terminals. The set $S$ that separates $T$ is called steiner cut. If it is of minimum possible size then it is called minimum steiner cut or minimum separator of set $T$. If $|T| \leq 1$, then $\kappa(T) = n-1$.

We define the \emph{terminal expansion} $h_T$ of a vertex cut $(L, S, R)$ in graph $G$ with respect to the terminal set $T$ as follows: \[h_T(L, S, R) = \frac{|S|}{\min\{|T\cap (L\cup S)|,|T\cap (R\cup S)|\}}.\]

We say the cut $(L, S, R)$ is $(T,\phi)$\emph{-expanding} in $G$ if $h_T(L, S, R)\geq \phi$ and $(T,\phi)$\emph{-sparse} if $h_T(L, S, R) < \phi$. We say that a graph $G$ is a $(T, \phi)$\emph{-expander} if every vertex cut $(L, S, R)$ such that $\min\{|T\cap (L\cup S)|,|T\cap (R\cup S)|\} >  0$  in the graph $G$ is $(T,\phi)$-expanding.

%% file: exact.tex
\section{An Exact Algorithm}\label{sec:exact}
In this section, we prove the main result stated in \Cref{thm:main}.  The main result follows from our terminal reduction algorithm. We recall \Cref{thm:main terminal reduction}.

\terminalreduction*

Recall that \Cref{thm:main} follows from \Cref{thm:main terminal reduction}. Indeed, starting as the entire vertex $V$ as a terminal set, we repeatedly applying \Cref{thm:main terminal reduction} for $O(\log n)$ times until we obtain a vertex cut of size less than $k$ or the terminal set becomes empty for which we know that the graph is $k$-connected. 

The rest of the section is devoted to proving \Cref{thm:main terminal reduction}. In \Cref{sec:unbalanced}, we show how to find a terminal-unbalanced mincut that has at most $\beta$ terminals on the smaller side, which will help to solve the vertex mincut fast when the graph is an expander in the base case of our algorithm. %
\Cref{sec:terminalRecursion} defines the concept of $k$-left and $k$-right graphs that helps us to implement the divide-and-conquer approach by using structural properties of the vertex mincut. In \Cref{sec:slowTerminalRedn}, we describe a slow terminal reduction algorithm based on the notion of $k$-left and $k$-right graphs. In \Cref{sec:terminal reduction}, we provide a fast implementation of the algorithm described in \Cref{sec:slowTerminalRedn} and prove the running time guarantee stated in \Cref{thm:main terminal reduction} using the algorithm in \Cref{sec:unbalanced} as a base case. %
\input{unbalancedVtxCut.tex}

\input{generalVtxCut.tex}
\input{terminalRecursion.tex}
\input{slowTerminalRedn.tex}
\input{fastTerminalRedn.tex}

%% file: unbalancedVtxCut.tex
\subsection{Unbalanced Vertex Cuts}\label{sec:unbalanced}
Let $G = (V, E)$ be a graph with terminal set $T \subseteq V$. We say that a vertex cut $(L,S,R)$ is $(T,\beta)$-\textit{unbalanced} if $\min\{ |T \cap (L \cup S)|, |T \cap (S \cup R)|\} < \beta$ and $(T,\beta)$-\textit{balanced} otherwise. Both sides of the cuts (including the separator) contain less than $\beta$ terminals. %
This section shows how to find such a cut if it exists.

\begin{theorem}\label{thm:unbalanced}
There is an algorithm, denoted as \textsc{Unbalanced}$(G, T,\beta)$, that takes as input a graph $G = (V, E)$, a terminal set $T \subseteq V$, and a parameter $\beta \geq 2$, and outputs a separator $S$ with the following guarantee: if $G$ contains a $(T,\beta)$-unbalanced vertex mincut, then $S_{\min}$ is a minimum separator. %

The algorithm makes calls to max-flows such that the total number of edges in all max-flow instances is at most $O(m\beta^2 \log^5 |T|)$, and the algorithm spends $O(m\beta^2 \log^4 |T|)$ time outside the max-flow calls.
\end{theorem}

To prove \Cref{thm:unbalanced}, we use two main ingredients described in \Cref{splitters} and \Cref{isolatingvertexcuts}. Below, we say that a vertex cut $(L,S,R)$ \emph{isolates} a terminal set $T$ if $|T\cap L|=1$, $|T\cap S| = 0$, $|T\cap R|\geq 1$. We refer to the vertex $v\in T\cap L$ as the \emph{isolated terminal}.

\begin{lemma}\label{splitters}
There exists an algorithm that, given parameters $t$ and $\beta \geq 2$ as input, constructs a family $\mathcal{T}$ of subsets of the terminal set $T$ where $|T| = t$, of size $|\mathcal{T}|\leq O(\beta^2\log^4 t)$ in time $O(\beta t\log^2 t)$ with the following guarantee:
for every $(T,\beta)$-unbalanced vertex cut $(L,S,R)$, there exists a terminal set $T'\in \mathcal{T}$ such that $(L,S,R)$ isolates $T'$.
\end{lemma}

The proof of \Cref{splitters} is deferred to the end of this section. The construction is relatively standard and based on the Hit and Miss hash family from \cite{s_deterministic_2023}. The second tool is the vertex version of the isolating cuts from \cite{li_vertex_2021}.

\begin{lemma}[Lemma 4.2 of \cite{li_vertex_2021}]\label{isolatingvertexcuts}
There exists an algorithm that takes an input graph $G=(V, E)$ and an independent set $I \subset V$ of size at least 2, and outputs for each $v \in I$, a $({v}, I\setminus {v})$-min-separator $C_v$. 
The algorithm makes $(s,t)$ max-flow calls on unit-vertex-capacity graphs with $O(m\log |I|)$ total number of vertices and edges and takes $O(m)$ additional time.
\end{lemma}

Given the two ingredients, we are now ready to prove our main theorem of the section.

\begin{proof}[Proof of \Cref{thm:unbalanced}]
The algorithm \textsc{Unbalanced}$(G, T,\beta)$ is as follows. First, construct a terminal set family $\mathcal{T}$ using \Cref{splitters} with $t=|T|,\beta$ as inputs. For every $T'\in \mathcal{T}$, compute an arbitrary maximal independent set $I_{T'}$ of $T'$ and apply \Cref{isolatingvertexcuts} on graph $G$ and terminal set $I_{T'}$ as inputs. Finally, we return $S_{\min}$ where $S_{\min}$ is the minimum separator over all the $({v}, I\setminus \{v\})$-min-separators $C_v$ for all $v \in I_{T'}$ and over all $T' \in \mathcal{T}$.

Now we prove the correctness of this algorithm. 
Let $(L, S, R)$ be the $(T,\beta)$-unbalanced vertex mincut in graph $G$.
From \Cref{splitters}, $(L,S,R)$ must isolate some terminal set $T' \in \mathcal{T}$. Let $v \in L\cap T'$ be the isolated terminal. Note that $v$ is not adjacent to any other terminal $v' \in T'$, so $v$ must be in the maximal independent set $I_{T'} \subseteq T'$. 
Since $S$ is a $(v, I_{T'}\setminus\{v\})$ min-separator, $|S_{\min}| \le |S|$ and we are done.

It remains to bound running time. From \Cref{splitters}, the time taken to construct $\mathcal{T}$ is $O(\beta t\log^2 t)$.
The time to compute the maximal independent sets $I_{T'}$ over all $T'\in \mathcal{T}$ is $O(m \beta^2\log^4 t)$ since $|\mathcal{T}|\leq O(\beta^2\log^4 t)$. Lastly, we analyze the total time to invoke \Cref{isolatingvertexcuts}. The algorithm makes max-flow calls on unit-vertex-capacity graphs with $O(\beta^2\log^4 t)\times O(m\log t) = O(m\beta^2\log^5 t)$ total number of vertices and edges and a total additional time of at most $O(\beta^2\log^4 t)\times O(m) = O(m\beta^2\log^4 t)$ outside max-flow calls.
\end{proof}

Hence, the rest of the section is devoted to proving \Cref{splitters}. We must introduce some formalism used in \cite{s_deterministic_2023} with slight modification for constructing such a terminal sets family.

\begin{definition}[Hit and Miss hash family]\label{hmhf}
Given $N,a,b,q,l,s\in \mathbb{N}$ such that $b\leq a$, we say that $\mathcal{H} \coloneqq \{h_i : [N] \rightarrow [q] \mid i \in [l]\}$ is a $[N,a,b,l,s]_q$-Hit and Miss (HM) hash family if
\begin{enumerate}
 \item For every pair of mutually disjoint subsets $A, B$ of $[N]$, where $|A|\leq a$ and $|B|\leq b$, there exists some $i\in [l]$ such that:
 \begin{center}
 $\forall (x,y)\in A \times B, h_i(x) \neq h_i(y)$
 \end{center}
 \item For every $h_i\in \mathcal{H}$ and every $j\in [q]$ we either have $h_{i}^{-1}(j) = \emptyset$ or $|h_{i}^{-1}(j)|\geq s$ where $h_{i}^{-1}(j) = \{ x \in [N] \mid h_i(x) = j\}$ denotes the set of all elements whose hash value is $j$ for the hash function $h_i$ and we call $s$ to be the minimum support size of the hash family $\mathcal{H}$.
\end{enumerate}
The computation time of a $[N, a, b, l, s]_q$ is defined to be the time needed to output the $l\times N$ matrix with entries in $[q]$ whose $(i,x)^{th}$ entry is $h_i(x)$ for $h_i \in \mathcal{H}$.
\end{definition}

Below, we show an HM-hash family with small $l$ and alphabet size $q$ but with large minimum support $s$. The construction is based on Lemma 16 of \cite{s_deterministic_2023}, but they did not analyze the minimum support size $s$ and the construction time.
Since both minimum support size $s$ and construction time are crucial for us, we give the proof for completeness.

\begin{claim}\label{primeshfwiths}
Given $N, a, b \in \mathbb{N}$ such that $b\leq a$ and $N\geq 200ab\log^2 N$, there exists a $[N, a, b, 1+ab\log N, \frac{N}{O(ab\log^2 N)}]_{O(ab\log^2 N)}$-HM hash family that can be constructed in time $O(abN\log N)$.
\end{claim}

\begin{proof}
We first state the construction of the HM-hash family from Lemma 16 of \cite{s_deterministic_2023} and then prove the bounds on parameters $l,q,s$ and running time.
Given $N, a, b$ as inputs, let $\mathcal{P}$ be the set of first $1+ab\log N$ prime numbers. The hash family $\mathcal{H}$ is defined as follows.
\[\mathcal{H} \coloneqq \{h_p: [N] \mapsto [q]\text{ with } h_p(x) = x (\mod p) \mid p\in \mathcal{P}\}.\]

Note that $p\leq q$ for all $p\in \mathcal{P}$.
From Lemma 16 of \cite{s_deterministic_2023}, we have the correctness guarantee that it hits and misses every pair of disjoint subsets $A, B$ of size $a, b$ respectively. It is also clear from above that $|\mathcal{H}| = l = 1+ab\log N$. Since the value of $(1+ab\log N)^{th}$ prime is at most $(1+ab\log N)\log{(1+ab\log N)} = O(ab\log^2 N)$, $q$ is bounded by $O(ab\log^2 N)$.

From \cite{h_primality_2002}, we know that it takes $\ot(\log^6 n)$ time to verify whether a number $n$ is prime. So to construct the set $\mathcal{P}$ it takes $O(ab\log^2 N)\times \ot(\log^6(O(ab\log^2 N))) = O(ab\log^2 N \log^6 a)$. Computing the hash value matrix takes $l\times N=O(abN\log N)$. Hence the total time to construct the hash family $\mathcal{H}$ is at most $O(abN\log N)$.

It remains to prove that the minimum support size $s$ is at least $N/O(ab\log^2 N)$. For every $h_p\in \mathcal{H}$ and $j\in [q]$, if $j \geq p$, then $h^{-1}(j)$ is empty. Otherwise, when $ j < p$, we have from the definition of $h_p$ that  $|h_p^{-1}(j)| \ge \lfloor{N/p}\rfloor \geq N/O(ab\log^2 N)$ as $p\leq 100ab\log^2 N$.
\end{proof}

From the guarantees given in \Cref{primeshfwiths}, we are now ready to construct our set family and prove \Cref{splitters}.

\begin{proof}[Proof of \Cref{splitters}]
Given $t,\beta$ as inputs, let $N=t, a = \beta-1, b=1$. Since $\beta \geq 2$ we have $b\leq a$. We can assume $t\geq 200\beta\log^2 t$. Otherwise, we can construct the terminal family $\mathcal{T}$ by choosing all possible pairs of the set $[t] = T$. Since there is at least one terminal in $L$ and $R$, one of the pairs is isolated by $(T,\beta)$-unbalanced cut $(L, S, R)$. The size of the family is at most $O(\beta^2\log^4 t)$ and the time taken to construct is at most $O(\beta^2\log^4 t)\leq O(\beta t\log^2 t)$.
\newline Since $b\leq a$ and $N\geq 200ab\log^2 N$, we can apply \Cref{primeshfwiths} and construct a HM hash family \[\mathcal{H} = [t, \beta-1,1,O(\beta \log t),\frac{t}{O(\beta\log^2 t)}]_{O(\beta\log^2 t)}.\] We construct the terminal set family $\mathcal{T}$ from $\mathcal{H}$ as follows.
\[\mathcal{T} = \{ h_p^{-1}(j) \mid h_p\in \mathcal{H}, j\in [p] \}.\]

The size of the terminal set family $\mathcal{T}$ is at most $|\mathcal{H}|\times q \leq \beta \log t \times \beta\log^2 t = \beta^2 \log^3 t$. For each hash function  $h_p \in \mathcal{H}$, 
we can construct all terminal sets $\{ h_p^{-1}(j) \mid j \in [p]\} \subseteq \mathcal{T}$ in one go by evaluating $h_p(x)$ over all $x \in [t]$. So it takes $t\times |\mathcal{H}| \leq O(\beta t \log t)$ time for the construction of $\mathcal{T}$. It remains to prove that for every $(T,\beta)$-unbalanced cut $(L,S,R)$, there exists a terminal set $T'\in \mathcal{T}$ such that $(L,S,R)$ isolates $T'$.

Let $(L,S,R)$ be a $(T,\beta)$-unbalanced cut with $|T\cap (L\cup S)|<\beta$.
Let $B=\{v\}$ where $v\in T\cap L$ is any arbitrary terminal and $A=(T \cap (L\cup S))\setminus \{v\}$. $A$ and $B$ are disjoint subsets of $T$ with $|B|=1$ and $|A|< \beta-1$. Based on the \Cref{hmhf}, there exists a hash function $h\in \mathcal{H}$ that hashes $v$ to a different value from all other terminals in $A$.
Let $j = h(v)$. Then, there exists a set $T^*=h^{-1}(j) \in \mathcal{T}$ that contains $v$ and does not contain any other terminal in $L\cup S$. From \Cref{primeshfwiths}, we have the guarantee that every non-empty set in $\mathcal{T}$ has size at least $s = t/100\beta\log^2 t \geq 2$ as $t\geq 200\beta\log^2 t$. Hence, we can conclude that $T^*\cap R \neq \emptyset$ and $(L,S,R)$ isolates $T^*$.
\end{proof}

%% file: terminalRecursion.tex
\subsection{Terminal Cut Recursion Lemma}\label{sec:terminalRecursion}
In this section, we define the key ideas that enable us to use the divide and conquer technique on graphs while preserving terminal vertex mincut. We first define in \Cref{def:left-right graphs} the notion of $k$-\textit{left} and $k$-\textit{right} graphs for a given cut parameter $k$ with respect to a vertex cut $(L, S, R)$ of $G$. We then prove in \Cref{lem:steiner leq k recursively} that the divide step preserves the terminal mincut of size less than $k$ and that the terminal mincut is recoverable from the smaller graphs in \Cref{lem:sep le k is sep in G}.

\begin{definition} \label{def:left-right graphs}
    Let $G = (V, E)$ be a graph with terminal set $T \subseteq V$. Let $(L, S, R)$ be a vertex cut in $G$. We define the $k$-\textit{left} graph $G_L$ (w.r.t $(L, S, R)$) of $G$ as the graph $G$ after replacing $R$ with a clique $K_R$  of size $k$ (selecting an arbitrary set of $k$ vertices that contains one terminal denoted as $t_R$), followed by adding biclique edges between $S$ and $K_R$. Optionally, we remove all edges inside $S$, see \Cref{fig:leftrightgraphs} for illustration. The definition of the $k$-right graph $G_R$ is similar where $K_L$ is a clique replacing $L$ with $t_L$ as a terminal in $K_L$.
\end{definition}

\begin{figure}[!ht]
\centering
\includegraphics[scale=0.35]{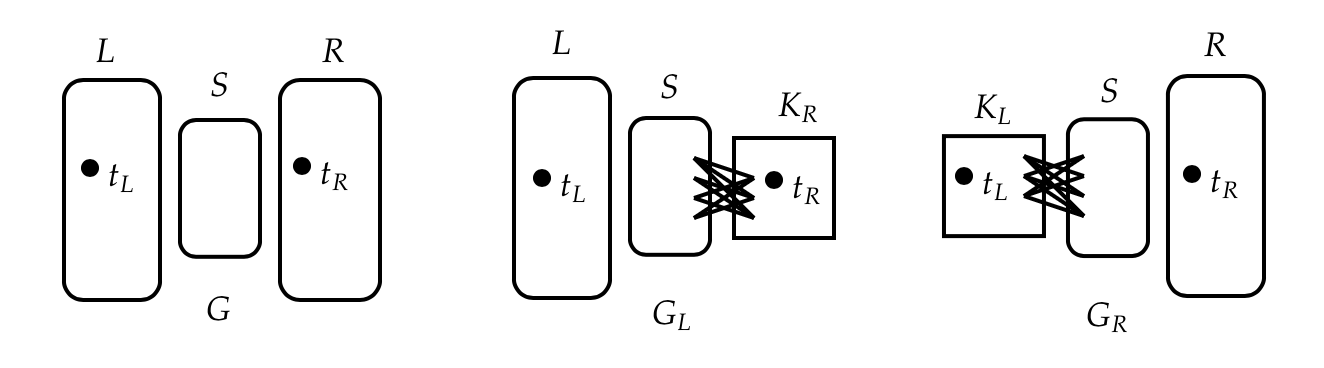}
\caption{$G_L, G_R$ are $k$-left, $k$-right graphs of $G$ w.r.t cut $(L, S, R)$.}
\label{fig:leftrightgraphs}
\end{figure}

Let $G$ be a graph with terminal set $T$. Let $(L,S,R)$ be a $T$-Steiner cut, and let $T_L := T \cap L$, $T_R := T \cap R$. Let $G_L$ be the $k$-left graph w.r.t. $(L,S,R)$ and denote $t_R$ to be an arbitrary terminal in the clique. We also define $G_R$ and $t_L$ symmetrically.

\begin{lemma} \label{lem:steiner leq k recursively}
If $\kappa_G(T) < k$, then $\min\{ \kappa_G(S), \kappa_{G_L}(T_L \cup  \{t_R\}), \kappa_{G_R}(T_R \cup \{t_L\})\} < k$.  
\end{lemma}

\begin{proof}
If $|S| < k$, then $\kappa_{G_L}(T_L \cup  \{t_R\}), \kappa_{G_R}(T_R \cup \{t_L\}) < k$ follows by \Cref{def:left-right graphs} and the fact that $(L,S,R)$ is a $T$-Steiner cut as shown in the left side of \Cref{fig:terminalrecursion1}. Now we assume $|S| \geq k$. Since $\kappa_G(T) < k$, there is a min Steiner cut $(L',S',R')$ of size less than $k$. If $\kappa_G(S) < k$, then we are done. Now assume $\kappa_G(S) \geq k$. Thus, $S'$ does not separate $x,y$ for all $x,y \in S$. That is, $S \subseteq L' \cup S'$ or $S \subseteq S' \cup R'$, i.e., $S \cap R' = \emptyset$ or $S \cap L' = \emptyset$. WLOG, we assume $S \cap R' = \emptyset$. Since $S'$ is a $T$-Steiner cut, $R' \cap T \neq \emptyset$, and so there must be a terminal vertex from $T$ in $L \cap R'$ or in $R \cap R'$. Let $t \in T$ be a terminal vertex in $L \cap R'$ (otherwise, the argument is  symmetric).

\begin{figure}[!ht]
    \centering
    \includegraphics[scale=0.9]{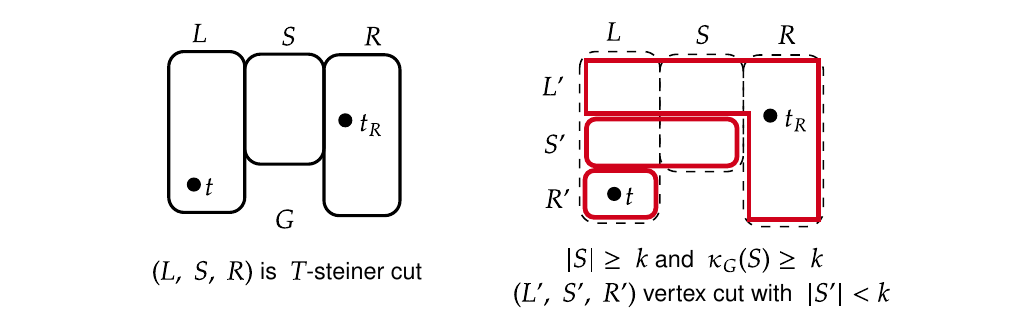}
    \caption{$S'$ separates $L\cap R'$ from the rest of the graph.}
    \label{fig:terminalrecursion1}
\end{figure}

Next, we prove that $S' \subseteq L \cup S$ (i.e., $S' \cap R = \emptyset$) as shown in the right side of \Cref{fig:terminalrecursion1}. Let $Z := N_G( L \cap R')$. Since $T \cap L' \neq \emptyset$ and $t \in L \cap R'$, $Z$ is a $T$-Steiner cut. Since  $S \cap R' = \emptyset$ and there is no edge from $L \cap R'$ to $R$ as there are no edges from $L$ to $R$, we have $Z =  S' \cap (L \cup S) \subseteq S'$.
Therefore, $|S'| = |S' \cap (L \cup S)| + |S' \cap R| = |Z| + |S'\cap R|$. Observe that $|S'| = |Z|$ since $Z$ is a $T$-Steiner cut where $Z \subseteq S'$ and $S'$ is a min $T$-Steiner cut. So,  $|S' \cap R| = 0$, and thus $S' \cap R = \emptyset$.

\begin{figure}[!ht]
    \centering
    \includegraphics[scale=0.35]{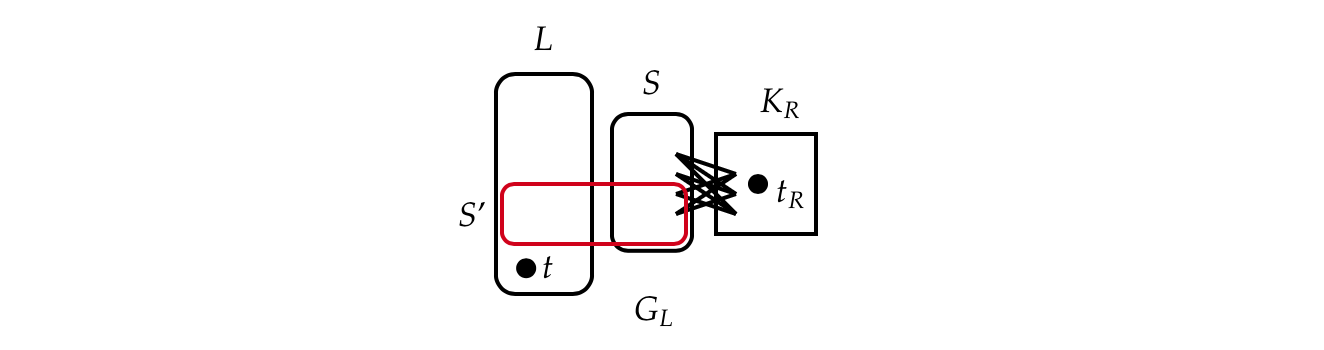}
    \caption{$S'$ is a $T_L\cup \{t_R\}$-separator in $G_L$}
    \label{fig:terminalrecursion2}
\end{figure}

Finally, we argue that $S'$ is a $T_L \cup \{t_R\}$-Steiner cut in $G_L$, refer \Cref{fig:terminalrecursion2}. Since $R' \cap S = \emptyset$ and $S' \cap R = \emptyset$, $S'$ remains a separator in $G_L$. It remains to prove that $S'$ separates $t$ and $t_R$ in $G_L$. Observe that $S'$ separates $t$ and $y$ in $G$ for all $y \in R$. In particular, $S'$ separates $t$ and $t_R$ in $G$. Therefore, $S'$ separates $t$ and $t_R$ in $G_L$. %
\end{proof}

\begin{lemma} \label{lem:sep le k is sep in G}
A separator of size less than $k$ in $G_L$ or $G_R$ is also a separator in $G$. Furthermore, for all $x,y \in V(G_L)$ (or in $V(G_R)$), a min $(x,y)$-separator of size less than $k$ in $G_L$ (or $G_R$)  is also an $(x,y)$-separator in $G$. %

\end{lemma}
\begin{proof}
We prove for $G_L$. The proof for $G_R$ is similar. Let $(L', S, R')$ be a vertex cut  smaller than $k$ in $G_L$. We prove that there is $S'' \subseteq S'$ such that $S''$ is a separator in $G$. If true, then $S'$ must also be a separator in $G$, and we are done. Since $|S'| < k$ and $K_R$ has $k$ vertices, there is a vertex $x \in K_R$ that is in $L'$ or in $R'$. WLOG, $x \in L'$. Since we have biclique edges between $K_R$ and $S$ in $G_L$ and there is no edge between $L'$ and $R'$, $S \subseteq L' \cup S'$ and $K_R \subseteq L'\cup S'$.
\begin{figure}[!ht]
    \centering
    \includegraphics[scale=0.9]{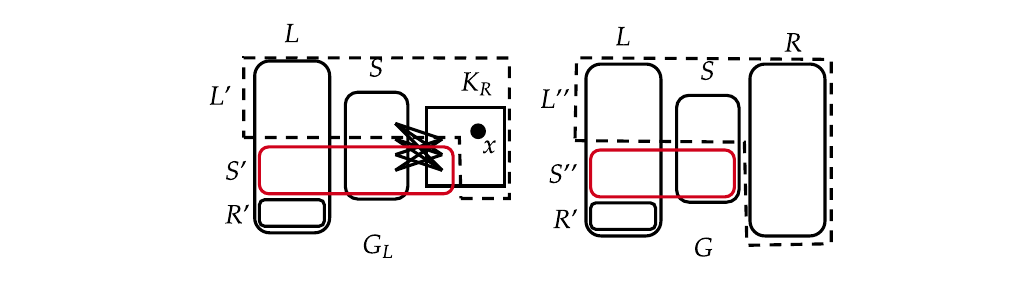}
    \caption{Moving $K_R$ from $L'$ to $L''$ followed by replacing $K_R$ with $R$}
    \label{fig:terminalrecursion3}
\end{figure}
By moving the clique to $L'$ as shown in \Cref{fig:terminalrecursion3}, $(L'',S'',R') := (L' \cup K_R, S'- K_R, R')$ is a vertex cut in $G_L$. To obtain $G$ from $G_L$, we replace $K_R$ with $R$, add edges between $R$ and $R \cup S$ according to $G$ and (possibly) add edges inside $S$. Since $K_R \subseteq L''$ and $S \subseteq L'' \cup S''$, we never add an edge from $L'' \cup R$ to $R'$. Therefore, $(L'' \cup R - K_R, S'', R')$ is a vertex cut in $G$, and so $S'' = S' - K_R \subseteq S'$ remains a separator in $G$.  The proof for the second statement is similar.
\end{proof}

%% file: slowTerminalRedn.tex
\subsection{Terminal Reduction Algorithm}\label{sec:slowTerminalRedn}

We now describe our algorithm that takes as inputs $G=(V, E), T\subseteq V$ and outputs either a separator of size less than $k$ or a terminal set $T'\subset V$ of size $|T'|<|T|/2$ such that $\kappa_G(T')<k$ whenever $\kappa_G(T)<k$. We only prove the correctness of \Cref{thm:main terminal reduction} using this algorithm and defer the implementation and proof of runtime guarantee to the next section.

\paragraph{Algorithm.} Let $k,\phi,\phibar$ be the global parameters such that $\phibar>\phi$. Given a graph $G$, if it is a $(T,\phi)$-expander, we can find mincut using \Cref{thm:unbalanced} as any cut of size less than $k$ is $(T,k/\phi)$-unbalanced. If $T$ has at most $10k/\phi$ terminals, we can run $(s,t)$-max-flow on all pairs of terminals to find the mincut.

So we can assume that there are at least $10k/\phi$ many terminals, and the graph has $(T,\phi)$-sparse cut. Let $(L, S, R)$ be a $(T,\phibar)$-sparse cut. We construct $G_L$ and $G_R$ w.r.t the vertex cut $(L, S, R)$ as defined in \Cref{def:left-right graphs} such that $t_R$ and $t_L$ are the terminals present in the clique $K_R$ and $K_L$ respectively and recurse on them. For $G_L$, we consider the terminals present in $L$ and the terminal $t_R$ in $K_R$ i.e. $(T\cap L)\cup \{t_R\}$, similarly, for $G_R$. Note that we ignore the terminals of $S$ in both cases.

If either of the recursive calls returns a separator of size less than $k$, we return that separator. Otherwise, they return terminal sets $S'_L$ and $S'_R$ such that $\kappa_{G_L}(S'_L)<k$ whenever $\kappa_{G_L}(T_L\cup \{t_R\})<k$ and $\kappa_{G_R}(S'_R)<k$ whenever $\kappa_{G_R}(T_R\cup \{t_L\})<k$. From \Cref{lem:steiner leq k recursively}, we have $\min\{ \kappa_G(S), \kappa_{G_L}(T_L \cup  \{t_R\}), \kappa_{G_R}(T_R \cup \{t_L\})\} < k$ whenever $\kappa_G(T)<k$. Hence we have $\min\{ \kappa_G(S), \kappa_{G_L}(S'_L), \kappa_{G_R}(S'_R)\} <k$ whenever $\kappa_G(T)<k$ and hence we return $T'=S\cup S'_L \cup S'_R$ as the new terminal set. 

We formally describe the algorithm in \Cref{alg:terminal reduction slow} and prove its correctness in the following lemma.

 \begin{algorithm}[H]
 \DontPrintSemicolon
 \KwGlobalVar{Connectivity parameter $k$ and $0.5 > \bar\phi > \phi > 0$.} %
 \KwIn{ A graph $G = (V,E)$ and a terminal set $T \subseteq V$, and connectivity parameter $k$ and $0.5 > \bar\phi > \phi > 0$.} 
 \KwOut{Either a separator of size less than $k$ or a new terminal set $T' \subseteq V$. }
 \BlankLine
 \If{$G$ \normalfont{is} a $(T,\phi)$-vertex expander or $|T| \leq 10k/\phi$ } {
 \Return{} $\emptyset$ if $\kappa_G(T) \geq k$. Otherwise, \Return{} a separator of size less than $k$ using all pairs max-flows or using \Cref{thm:unbalanced}.
 }
 Let $(L,S,R)$ be a $(T,\bar \phi)$-sparse cut. \;
 Let $G_L$ and $G_R$ be $k$-left and $k$-right graphs of $G$
 w.r.t. $(L,S,R)$.\;
 Let $t_R$ be a terminal in the clique of $G_L$, and $t_L$ be a terminal in the clique of $G_R$.\; 
 Let $T_L = T \cap L$, and $T_R = T \cap R$. \;
 $S'_L \gets \textsc{ReduceTerminalSlow}_{k,\phi,\bar \phi}(G_L,T_L \cup \{t_R\})$ \;
 $S'_R \gets \textsc{ReduceTerminalSlow}_{k,\phi,\bar \phi}(G_R,T_R \cup \{t_L\})$ \;
 \If{$S'_L$ \normalfont{(or} $S'_R$) \normalfont{is} a separator in $G_L$ (or $G_R$) of size $<k$}{
 By \Cref{lem:sep le k is sep in G}, it must be a separator in $G$. \;
 \textbf{return} the corresponding separator.} 
 \Else{
 $S'_L$ and $S'_R$ are both new terminal sets of $G$. \;
 \Return{ $S_L' \cup S'_R \cup S.$}
 }
\caption{\textsc{ReduceTerminalSlow}$_{k,\phi,\bar \phi}(G,T)$}
\label{alg:terminal reduction slow}
\end{algorithm}

\begin{lemma} \label{lem:slow alg terminal correct}
\Cref{alg:terminal reduction slow} on $G,T$ with parameters $k,\phi, \bar \phi$ outputs either a separator of size $< k$ or a new terminal set $T' \subseteq V$ such that $\kappa(T') < k$ whenever $\kappa(T) < k$. 
\end{lemma}
Observe that the new terminal set $T'$ size depends on the recursion depth, which we bound in the next section.%
\begin{proof}
We prove that \Cref{alg:terminal reduction slow} outputs correctly as stated on every connected graph with a terminal set by induction on the number of terminals $t$. Observe that the parameters $k,\phi,\bar \phi$ are fixed. When $t \leq 10k\phi^{-1}$, it is the base case and we handle it by computing all pairs max-flow over the terminal set $T$. Next, for all $t \geq 10k\phi^{-1}+1$, we assume that the statement holds up to $t-1$ terminals, and we now prove that the statement continues to hold for $t$ terminals. If $G$ is a $(T,\phi)$-expander, then we know that $G$ has a $(T,k/\phi)$-unbalanced cut and we can use \Cref{thm:unbalanced} in this case to detect a cut of size less than $k$. Otherwise, let $(L,S,R)$ be a $(T,\phibar)$-sparse cut. Denote $T_L = T \cap L, T_R = T \cap R, T_S = T \cap S$. We prove that $|T_L|, |T_R| \geq 2$. If true, then $G_L$ and $G_R$ will have $|T_L \cup \{t_R\}| \leq t-1$ and $|T_R \cap \{t_L\}| \leq t-1$ terminals, respectively, so that we can apply inductive hypothesis on them. Suppose $|T_L| = 1$. Then, the expansion $h(L,S,R) = |S|/(1+|T_S|)$. If $T_S = \emptyset$, then $h(L,S,R) \geq 1 > \bar\phi$. Otherwise, $h(L,S,R) = \frac{|S|}{1+|T_S|} \geq \frac{|T_S|}{1+|T_S|} \geq 1/2 > \bar \phi$. In either case, we have a contradiction.

It remains to prove that \Cref{alg:terminal reduction slow} on $G$ and terminal set $T$ outputs correctly as stated in the lemma. Let $S'_L$ and $S'_R$ be the output from $\textsc{ReduceTerminalSlow}_{k,\phi,\bar \phi}(G_L,T_L \cup \{t_R\})$ and $\textsc{ReduceTerminalSlow}_{k,\phi,\bar \phi}(G_R,T_R \cup \{t_L\})$, respectively. By inductive hypothesis, $S'_L$ is either a separator in $G_L$ of size $< k$ or a new terminal set such that $\kappa_{G_L}(S'_L) < k$ whenever $\kappa_G(T_L \cup \{t_R\}) < k$. Similarly, $S'_R$ is either a separator in $G_R$ of size $< k$ or a new terminal set such that $\kappa_{G_R}(S'_R) < k$ whenever $\kappa_G(T_R \cup \{t_L\}) < k$. If $S'_L$ or $S'_R$ is a separator for $G_L$ or $G_R$, respectively, then \Cref{lem:sep le k is sep in G} implies that one of them is a separator in $G$, and we are done. Now, assume that $S'_L$ and $S'_R$ are the new terminal sets. We prove that $\kappa_G(S'_L \cup S'_R \cup S) < k$ whenever $\kappa_G(T) < k$. Let us assume that $\kappa_G(T) < k$. We are done if $\kappa_G(S) < k$. Now, assume $\kappa_G(S) \geq k$. By \Cref{lem:steiner leq k recursively}, either $ \kappa_{G_L}(T_L \cup \{t_R\}) < k$ or $\kappa_{G_R}(T_R \cup \{t_L\}) < k$, and thus $\kappa_{G_L}(S'_L) < k$ or $\kappa_{G_R}(S'_R) <k$. By \Cref{lem:sep le k is sep in G}, $\kappa_{G}(S'_L) < k$ or $\kappa_{G}(S'_R) <k$, and we are done.%
\end{proof}

%% file: fastTerminalRedn.tex
\subsection{Fast Implementation of \texorpdfstring{\Cref{alg:terminal reduction slow}}{terminal reduction}} \label{sec:terminal reduction}

We have proved the correctness of \Cref{alg:terminal reduction slow} (\textsc{ReduceTerminalSlow}) in the previous section. In this section, we use the following important lemma from \cite{long_near-optimal_2022} to implement \textsc{ReduceTerminal} following the approach in \Cref{alg:terminal reduction slow}.

\begin{lemma} [\cite{long_near-optimal_2022}] \label{lem:balancedorexpander}
Let $G=(V,E)$ be an $n$-vertex $m$-edge graph with a terminal set $T \subseteq V$. Given a parameter
$0<\phibar\le1/10$ and $1\le r\le\left\lfloor \log_{20}n\right\rfloor $,
there is a deterministic algorithm, $\textsc{TerminalBalancedCutOrExpander}(G,T, \phibar,r)$, that computes a vertex cut $(L,S,R)$
(but possibly $L=S=\emptyset$) of $G$ (where we denote $T_{L \cup S} := T \cap (L \cup S)$ and $T_{R \cup S} := T \cap (R \cup S)$) such that $|S|\le\phibar|T_{L\cup S}|$ which further satisfies 
\begin{itemize}
\item either $|T_{L\cup S}|,|T_{R\cup S}|\ge |T|/3$; or
\item $|T_{R\cup S}|\ge |T|/2$ and $G[R\cup S]$ is a $(T_{R\cup S},\phi)$-vertex expander
for some $\phi=\phibar/\log^{O(r^{5})}m$.
\end{itemize}
The running time of this algorithm is $O(m^{1+o(1)+O(1/r)}\log^{O(r^{4})}(m)/\phi)$. %
\end{lemma}

\paragraph{Algorithm.} $\textsc{ReduceTerminal}_{k,\phi}$ takes $G=(V, E)$ and $T\subseteq V$ as inputs and outputs either a separator of size less than $k$ or a new terminal set of a smaller size where $k,\phi$ are the global parameters.

If the graph has a small number of terminals, say at most $10k/\phi$, then we run max-flow on all pairs of the terminals, and we either return the mincut or return an empty terminal set denoting that $\kappa_G(T)\geq k$. Let $r,\phibar$ be parameters as defined in the \Cref{alg:fast expander or terminal family} which are given as inputs to \textsc{TerminalBalancedCutOrExpander} along with $G, T$ that outputs a cut $(L, S, R)$.

From the \Cref{lem:balancedorexpander} guarantee, the graph $G$ is a $(T,\phi)$-expander if $L=S=\emptyset$. Based on the definition of the expander, there are at most $\Theta(k/\phi)$ number of terminals on the smaller side of the mincut; hence it is $(T, \beta)$-unbalanced for $\beta=\Theta(k/\phi)$ and we use \Cref{thm:unbalanced} (\textsc{Unbalanced}) to solve for a cut of size less than $k$.

Hence, $L\neq S\neq \emptyset$ and we get a $(T,\phibar)$-sparse cut by \Cref{lem:balancedorexpander}. If $|S|<k$, we return $S$ as the separator. Otherwise, we have the two cases mentioned in the \Cref{lem:balancedorexpander}. In both cases, we must recurse on the left side of the cut. Hence, we construct the $k$-left graph $G_L$ as defined in \Cref{def:left-right graphs} and apply recursively $\textsc{ReduceTerminal}_{k,\phi}$ with $G_L,( T\cap L)\cup\{t_R\}$ as inputs where $t_R$ is any arbitrary terminal in $R$. If the output $S'_L$ is a separator, then we return it.

Otherwise, we have to recurse on the right side, so we construct $G_R$ similarly. If we are in the first case of \Cref{lem:balancedorexpander}, then we simply call $\textsc{ReduceTerminal}_{k,\phi}$ with the terminal set $(T\cap R)\cup \{t_L\}$, where $t_L$ is arbitrary terminal in $R$. If the output $S'_R$ is a separator, then we return it. Otherwise, $S'_R$ is the terminal set that we need to recurse.

If we are in the second case of \Cref{lem:balancedorexpander}, then we have the guarantee that $G[S \cup R]$ is a $(T_{R \cup S},2\phi)$-vertex expander. We prove that $k$-right graph (w.r.t. $(L,S,R)$) $G_R$ is still an $(T_{R \cup S} \cup \{t_L\},\phi)$-vertex expander. Note in this case we do not remove the edges inside $S$ while constructing $G_R$.

\begin{claim} \label{claim:GR still a expander}
 If $G[S \cup R]$ is a $(T_{R \cup S},2\phi)$-vertex expander, then the $k$-right
 graph $G_R$ is a $(T_{R \cup S} \cup \{t_L\},\phi)$-vertex expander. %
\end{claim}
\begin{proof}
 Suppose there is a $(T_{R\cup S} \cup \{t_L\},\phi)$-sparse cut $(L',S',R')$ in $G_R$. By
 removing the clique in $G_R$, we obtain another cut $(L'',S'',R'')$
 in $G[S\cup R]$ where $L'' = L' -V(K_L)$ $S'' = S' - V(K_L)$ and $R'' = R -
 V(K_L)$. The expansion $h(L'',S'',R'')$ is \[ \frac{|S''|}{|T_{R\cup S} \cap (L'' \cup S'')|}
 \leq \frac{|S'|}{|T_{R\cup S} \cap (L' \cup S')|} \leq \frac{|S'|}{|(T_{R\cup S} \cup \{t_L\}) \cap (L' \cup S')|-1} < 2\phi.\] The first inequality follows since $|T_{R\cup S} \cap (L'' \cup S'')| = |T_{R\cup S} \cap (L' \cup S')|$. The last inequality follows since $(L',S',R')$ is a $(T_{R\cup S} \cup \{t_L\},\phi)$-sparse, and thus $|(T_{R\cup S} \cup \{t_L\}) \cap (L' \cup S')| \geq 2$. So, $(L'',S'',R'')$ is $(T_{R \cup S},2\phi)$-sparse in $G[S\cup R]$, a contradiction.
\end{proof}

Hence, we can just use \textsc{Unbalanced} on $G_R$ with $(T\cap (R\cup S)) \cup \{t_L\}$ as the terminal set and $\beta = \Theta(k/\phi)$. If it returns a separator, we return it. Otherwise, we guarantee no terminal cut of size less than $k$ in $G_R$; hence, we need not recurse further so $S'_R=\emptyset$. Finally, we return the union of $S'_L, S, S'_R$ as the new terminal set.

We describe the algorithm formally in \Cref{alg:fast expander or terminal family} and prove the correctness and running time guarantee in \Cref{lem:fast terminal reduction}.

\begin{algorithm}[H]
 \DontPrintSemicolon
 \KwGlobalVar{the original input graph $G^{\textrm{orig}} = (V^{\textrm{orig}},E^{\textrm{orig}})$ whose min degree is $\geq k$. The parameters $\phi = (1/\log^{(\log \log |V^{\textrm{orig}}|)^6}|E^{\textrm{orig}}|) \in (0,1/10)$ and $k$ are also global.}
 
 \KwIn{ A graph $G = (V,E)$ with terminal set $T \subseteq V$.} %
 \KwOut{A separator of size $<k$ or a new terminal set.}%
 \BlankLine
 $r \gets \log \log |V|$, $\phibar \gets 2\cdot \phi \cdot \log^{O(r^5)}|E|$\;
 \If{$|T| \leq 10k/\phi$ \label{line:Tle2k}}{ 
 Let $S'$ be a minimizer of $\min_{x,y \in T}\kappa_G(x,y)$. \;
 \Return $S'$ if $|S'| < k$, otherwise \Return an empty (terminal) set. \;
 }
 $(L,S,R) \gets \textsc{TerminalBalancedCutOrExpander}(G, \phibar,r)$\;
 \If{$L = S = \emptyset$}{ \label{line:L=S=0}
 $S' \gets \textsc{Unbalanced}(G,T,\beta)$ where $\beta = \Theta(k\phi^{-1})$ using \Cref{thm:unbalanced}. \;
 \Return $S'$ if $|S'| < k$, otherwise \Return an empty (terminal) set. \;
 }
 \lIf{$|S| < k$}{ \label{line:sparse is mincut}
 \Return{$S$}
 }
 Let $T_L = T \cap L, T_{L\cup S} = T \cap (L \cup S), T_R = T \cap R$ and $T_{R \cup S} = T \cap (R \cup S)$. \tcp*{ By \Cref{lem:balancedorexpander}, $(L,S,R)$ is $(T,\phibar)$-sparse.}
 Let $G_L$ be the $k$-left graph of $G$ w.r.t. $(L,S,R)$, and $t_R$ be a terminal in the clique of $G_L$ where we remove all edges in $E_{G_L}(S,S)$ from $G_L$. \;

 $S'_L \gets \textsc{ReduceTerminal}_{\phi,k}(G_L,T_L \cup \{t_R\})$ \;
 \If{$S'_L$ \normalfont{is} a separator} {\Return{$S'_L$} \tcp*{$|S'_R| < k$}}

 Let $G_R$ be the $k$-right graph of $G$ w.r.t. $(L,S,R)$, and $t_L$ be a terminal in the clique of $G_R$ where we remove all edges in $E_{G_R}(S,S)$ from $G_R$ if and only if $\min\{T_{L \cup S}, T_{R\cup S}\} \geq |T|/3$. \;

 \If{$\min\{|T_{L \cup S}|, |T_{R \cup S}| \} \geq |T|/3$} {
 $S'_R \gets \textsc{ReduceTerminal}_{\phi,k}(G_R,T_R \cup \{t_L\})$\;
 \If{$S'_R$ \normalfont{is} a separator}
 {\Return{$S'_R$}\tcp*{$|S'_R| < k$}}
 }\Else{
 By \Cref{claim:GR still a expander}, $G_R$ is a $(T_{R\cup S}\cup \{t_L\}, \phi)$-vertex expander.\;
 $S'_R \gets \textsc{Unbalanced}(G_R,T_{R\cup S} \cup \{t_L\},\beta) \text{ where } \beta = \Theta(k\phi^{-1})$ using \Cref{thm:unbalanced}. \label{line:unbal gr expander}\;
 \If{$S'_R$ \normalfont{is} a separator of size $<k$}
 {\Return{$S'_R$}}
 \Else{$S'_R \gets \emptyset$\;}%

 }
 \tcp*{$S'_L$ and $S'_R$ are new terminal sets.}
 \Return{$S'_L \cup S \cup S'_R$} as a new terminal set.

\caption{\textsc{ReduceTerminal}$_{\phi,k}(G,T)$}
\label{alg:fast expander or terminal family}
\end{algorithm}

\begin{lemma}\label{lem:fast terminal reduction}
 Let $G= (V,E)$ be a graph with a terminal set $T \subseteq V$ where $G$ has min-degree at least $k$ and $|V| = n, |E| = m, \phi \in (0,1/10)$. The algorithm \textsc{ReduceTerminal}$_{\phi,k}(G,T)$ (\Cref{alg:fast expander or terminal family}) either returns a separator of size $<k$ or a new terminal set $T' \subseteq V$ such that $\kappa(T') < k$ whenever $\kappa(T) < k$. Furthermore, $|T'| \leq \phibar|T|(1+\phi)^{O(\log n)}$ where $\phibar = \phi \cdot m^{o(1)}$. The algorithm runs $O(m'^{1+o(1)}\phi^{-1})$ time outside max-flow calls, and the total number of edges in all max-flow instances is $O(m'k^2 \phi^{-2}\log^2 n)$ where $m' = (1+4\phibar)^{O(\log n)}m$.
\end{lemma}

The main theorem \Cref{thm:main terminal reduction} is implied by substituting an appropriate value of $\phi$ as given in \Cref{alg:fast expander or terminal family}. The rest of the section is devoted to proving the \Cref{lem:fast terminal reduction}.

\paragraph{Correctness.} We show that \Cref{alg:fast expander or terminal family} follows the description of \Cref{alg:terminal reduction slow}, i.e., it implements \Cref{alg:terminal reduction slow}. By \Cref{lem:slow alg terminal correct}, this means that \Cref{alg:fast expander or terminal family} outputs either a separator of size less than $k$ or a new terminal set $T' \subseteq V$ such that $\kappa(T') < k$ whenever $\kappa(T) < k$. The base cases happen when $|T| \leq 10k/\phi$ (\Cref{line:Tle2k}) or when $L = S = \emptyset$ (\Cref{line:L=S=0}). The correctness is immediate if $|T| \leq 10k/\phi$. If $L = S = \emptyset$ at \Cref{line:L=S=0}, then $G$ is a $(T,2\phi)$-vertex expander, so $G$ contains a $(T,\beta)$-unbalanced vertex mincut whenever $\kappa_G(T) <k$. By \Cref{thm:unbalanced}, \textsc{Unbalanced}$(G,T,\beta)$ outputs a separator $S'$ of size less than $k$ whenever $\kappa_G(T) < k$. If $|S'| \geq k$, we have verified that $\kappa_G(T) \geq k$, and thus an empty terminal set is returned correctly. If $|S| < k$ (\Cref{line:sparse is mincut}), then $S$ must be a separator of $G$, and we are done.

We now assume that we do not execute the base cases. By \Cref{lem:balancedorexpander}, $(L,S,R)$ is a $(T,\phibar)$-sparse. We recurse on $(G_L,T_L \cup \{t_R\})$ in the same way as in \Cref{alg:terminal reduction slow}. For $G_R$, either we recurse on $(G_R,T_R\cup \{t_L\})$ or we immediately solve $G_R$. If $\min\{|T_{L \cup S}|, |T_{R \cup S}| \} \geq |T|/3$, then we recurse on $(G_R,T_R\cup \{t_L\})$ in the same way as in \Cref{alg:terminal reduction slow}. Otherwise, \Cref{lem:balancedorexpander} together with \Cref{claim:GR still a expander} imply that $G_R$ is $(T_{R\cup S} \cup \{t_L\},\phi)$-vertex expander. Thus, $G_R$ contains a $(T_{R \cup S} \cup \{t_L\},\beta)$-unbalanced vertex mincut whenever $\kappa_{G_R}(T_{R\cup S}\cup\{t_L\}) < k$. By \Cref{thm:unbalanced}, \textsc{Unbalanced}$(G_R,T_{R\cup S} \cup \{t_L\},\beta)$ (\Cref{line:unbal gr expander}) outputs a separator $S'_R$ of size $<k$ whenever $\kappa_{G_R}(T_{R\cup S}\cup\{t_L\}) < k$. If $|S'_R| \geq k$, then $\kappa_{G_R}(T_{R\cup S} \cup \{T_L\}) \geq k$, and thus we can view the subproblem on $G_R$ as returning $S'_R$ as an emptyset, and hence the new terminal set $S'_L \cup S \cup S'_R$ is returned according to \Cref{alg:terminal reduction slow}. %

\paragraph{Recursion tree.} It remains to analyze the size of the new terminal and the total running time. We first set up notations regarding the recursion tree $\mathcal{T}$. Let $G^{\textrm{orig}}, T^{\textrm{orig}}$ be the graph and its terminal set at the root of the recursion tree $\textsc{ReduceTerminal}_{\phi,k}(G^{\text{orig}}, T^{\text{orig}})$. We represent each subproblem by $(G, T)$, the graph and its terminal set. If $G$ is $(T,\phi)$-vertex expander or $|T| \le 10k/\phi$, then we treat this subproblem as a leaf node in the recursion tree. Otherwise, we obtain a $(T,\phibar)$-sparse cut $(L,S,R)$, and recurse on $(G_L,T_L \cup \{t_R\})$ and $(G_R,T_R \cup \{t_R\})$. Here, we treat $(G_L, T_L \cup \{t_R\})$ and $(G_R,T_R \cup \{t_L\})$ as a left child and a right child -- respectively of $(G,T)$. Furthermore, for each internal node $(G, T)$, we also denote $(G, T; (L, S, R))$ where $(L, S, R)$ is the sparse cut obtained in the subproblem. For the analysis, let us assume that \Cref{line:sparse is mincut} is always false. Each subproblem's sparse cut has size at least $k$, which can only give a worse bound.

\begin{lemma} \label{lem:recursion depth log n}
The depth of the recursion tree $\mathcal{T}$ is $O(\log |T^{\textrm{orig}}|) = O(\log n)$.
\end{lemma}
\begin{proof}
Let $(G, T)$ be any internal node in the recursion tree. Let $(L, S, R)$ be $(T,\phibar)$-sparse cut in $G$ with the properties as stated in \Cref{lem:balancedorexpander}. We show that the terminal set of each subproblem decreases by a constant factor. If $|T_{L\cup S}|, |T_{R\cup S}| \geq |T|/3$, then $|T_L|, |T_R| \leq |T| - |T|/3 \leq 2|T|/3$. Therefore, $|T_L \cup \{t_R\}|, |T_R \cup \{t_L\}| \leq 2|T|/3 + 1 \leq 0.77|T|$. The last inequality follows since $|T| \geq 10k/\phi$. Otherwise, $G_R$ is $(T_R \cup \{t_L\}, \phi)$-vertex expander, and we treat the right subproblem as a leaf in the recursion tree. In this case, $|T_L| = |T| - |T_{R\cup S}| \leq |T|/2$, and thus $|T_L \cup \{t_R\}| = |T_L| + 1 \leq 0.6|T|$. %
\end{proof}

\paragraph{The size of the new terminal set.} If $\textsc{ReduceTerminal}_{\phi,k}(G^{\text{orig}},T^{\text{orig}})$ returns a separator, then it must be of size less than $k$, and we are done. Now we assume that a terminal set $T'$ is returned where we denote $t = |T^{\text{orig}}|, t' = |T'|$. Thus, every subproblem in the recursion must return a (possibly empty) terminal set. Let $(G, T; (L, S, R))$ be an internal node in the recursion tree. Let $G_L$ and $G_R$ be the $k$-left/-right graphs of $G$ w.r.t. $(L,S,R)$, respectively. Denote $t_0 = |T|$, we have $|T_L \cup \{t_R\}| +|T_R \cup \{t_L\}| \leq t_0+2 \leq (1+\phi)t_0$. The last inequality follows since $|T| \geq 10k\phi^{-1}\geq 10\phi^{-1}$ for non-base cases. Therefore, at level, $i$, the total number of terminals (as inputs) is at most $(1+\phi)^{i-1}t$. Let $\ell$ be the recursion depth. The total number of terminals as inputs to each subproblem at all levels is $O((1+\phi)^{\ell+1} t)$. Next, we bound the size of the output terminals. Observe that, every internal node $(G,T;(L,S,R))$, $(L,S,R)$ is a $(T,\phibar)$-sparse cut, and thus $|S| \leq \phibar|T_{L\cup S}| \leq \phibar \cdot t_0$. Therefore, the number of output terminals is at most $\phibar$ times the total number of input terminals. Thus, $|T'| = t' \leq \phibar(1+\phi)^{\ell+1}|T| \leq \phibar(1+\phi)^{O(\log n)}|T|$. The last inequality follows by \Cref{lem:recursion depth log n}.

\paragraph{Running time.} We bound the total number of edges of the graphs (as an input to each subproblem) in the recursion tree. Let $(G, T;(L, S, R))$ be an internal node in the recursion tree. We assume that $\min\{T_{L\cup S}, T_{R\cup S}\} \geq |T|/3$ because otherwise we would recurse only on $G_L$. We denote $n_0 = |V(G)|, m_0 = |E(G)|, m_L = |E(G_L)|, m_R = |E(G_R)|$. Therefore,
\begin{align*}
 m_L + m_R \leq m_0 + 2(|S|k+k^2)
 \leq m_0 + 4|S| k
 \leq m_0 + 4 \phibar \cdot n_0 k
 \leq m_0(1+4\phibar)
\end{align*}
The first inequality follows since there are $|S|k +k^2$ additional edges from biclique edges in $G_L$ and $G_R$ and the clique edges, and we remove edges between $S$ in both $G_L$ and $G_R$. The second inequality follows since $|S| \geq k$. The third inequality follows since $|S| \leq \phibar\cdot |T_{L\cup S}| \leq \phibar n_0$. The last inequality follows since the min-degrees of $G_L$ and $G_R$ are at least $k$.

Therefore, at level $i$, the total number of edges is at most $O(m(1+4\phibar)^{i-1})$. By summing over all $i$, the total number of edges in the recursion tree is $O(m(1+4\phibar)^{\ell+1}) = O(m(1+4\phibar)^{O(\log n)})$. Finally, we bound the running time using \Cref{thm:unbalanced} at the base cases and \Cref{lem:balancedorexpander} at each internal node of the recursion tree. At the base cases, every instance of $m'$ edges takes $O(m'\beta^2\log^4(n)) = \ot(m'k^2\phi^{-2})$ time outside the max-flows, and the number of edges in all max-flow instances is at most $O(m'\beta^2\log^5(n)) = O(m'k^2\phi^{-2}\log^5(n))$. At each internal node in the recursion, each instance of $m'$ edges takes $m'^{1+o(1)}\phi^{-1}$. %

%% file: approxImprov.tex
\section{A \texorpdfstring{$(1+\eps)$}{(1+eps)}-Approximation Algorithm}\label{sec:approx}

In this section, we give a $(1+\eps)$-approximation algorithm for vertex connectivity. Recall \Cref{thm:approx} from \Cref{sec:intro}.

\approximationthm*

The main idea of the algorithm is to use two graphs\footnote{the first graph is an induced sub-graph of an expander and the second graph is small set vertex expander.} and apply $(s,t)$-min cut on all the edges of them. We start with defining the guarantees required from these graphs and then state and prove the correctness of the algorithm.

\begin{lemma}[Induced sub-graph of Ramanujan expander]\label{lem:ramanujan}
Given an integer $n$ and a degree parameter $d$, a graph $H=(V_H,E_H)$ with maximum degree $4d$ and $|V_H|=n$ can be constructed in time $O(nd)$ such that for any $L,R\subseteq V_H$, if $|L|\cdot |R| \cdot d \geq 4\cGabow^2n^2$ for universal constant $\cGabow$ from \Cref{lem:gabow00}, then $E_H(L,R)\neq \emptyset$.
\end{lemma}

A graph $X = (V_X, E_X)$ is called $(\alpha, \beta)$-vertex expander for some $\alpha\leq 1, \beta \ge 0$, if for every subset $L \subset V_X$ with $|L| \leq \alpha |V_X|$ we have $|N_X(L)|\geq \beta |L|$ where $N_X(L) = \{u \in V\setminus L \mid \exists v \in L \text{ such that } (u,v)\in E_X\}$ represents the neighborhood of $L$ in $X$.

\begin{lemma}[Small set vertex expander]\label{lem:vertexExpander}
Given an integer $n$ and a  parameter $\eps$, a $(\frac{\eps}{20c^2},\frac{2}{\eps})$-vertex expander, $H=(V_H, E_H)$ with $|V_H|=n$ and maximum degree $O(1/\eps)$ can be constructed in $O(n/\eps)$ time.
\end{lemma}

Before proving \Cref{thm:approx}, we must observe a structural property of the $(1+\eps)$-approximate minimum cut. Let $\delta$ be the minimum degree of the graph $G$.

\begin{claim}\label{lem:epskbalanced}
Let $(L,S,R)$ be any minimum vertex cut of size $\kappa_G$ in graph $G=(V,E)$ with minimum degree $\delta\geq (1+\eps)\kappa_G$. We have $\min(|L|,|R|) \geq \eps \kappa_G$.
\end{claim}

\begin{proof}
WLOG, we can assume $|L|\leq |R|$ and $v$ be any arbitrary vertex in $L$. Since $(L,S,R)$ is a vertex cut, $N(v) \subseteq L\cup S\setminus \{v\}$. So we have \[(1+\eps)\kappa_G \leq \delta \leq |N(v)| \leq |L|+|S|-1 = |L|+\kappa_G-1\] which implies $|L|\geq \eps \kappa_G$.
\end{proof}

The above lemma tells us that if the minimum degree of graph $\delta \geq (1+\eps)\kappa_G$, every minimum cut has at least $\eps \kappa_G$ vertices on both sides. We are now ready to state our algorithm and prove its correctness and running time.

\begin{algorithm}[H]
 \DontPrintSemicolon
 \KwIn{A graph $G=(V,E)$, approximation parameter $\eps$}
 \KwOut{A separator $S$ of graph $G$ with guarantee of $|S|<\floor{(1+\eps)\kappa_G}$}
 \BlankLine
 Let $\cGabow$ be the universal constant given in \Cref{lem:gabow00}.\\
 Let $H_1 = (V,E_{H_1})$ be the graph constructed with $|V|,d=\frac{1600\cGabow^6}{\eps^2}$ as inputs to \Cref{lem:ramanujan}\\
 Let $H_2 = (V,E_{H_2})$ be the $(\frac{\eps}{20\cGabow^2},\frac{2}{\eps})$-vertex expander constructed with $|V|,\eps$ as inputs to \Cref{lem:vertexExpander}\\
 Let $S$ be the minimum separator seen so far, initialized with $V$.\\
 \For{$e=(s,t) \in E_{H_1}\cup E_{H_2}$\label{expanderedges}}
 {
 Compute $T =$  min $(s,t)$-separator in $G$.\;
 \If{$|T|<|S|$}
 {$S = T$}
 }
 \eIf{$\delta < |S|$}
 {\Return $N(v)$ where $v\in V$ such that $\deg(v)=\delta$.\label{alg:approx:1}}
 {\Return $S$.}
 \caption{\textsc{ApproxVertexMinCut}$(G,\eps)$}
 \label{alg:approx}
\end{algorithm}

We prove the correctness and running time of \Cref{alg:approx}, thus proving \Cref{thm:approx}.

\begin{proof}[Proof of \Cref{thm:approx}]
Let $(L, S, R)$ be a minimum vertex cut in the graph $G$ with $|V|=n$. We construct $H_1,H_2$ by giving required inputs to \Cref{lem:ramanujan,lem:vertexExpander} respectively. Since the three graphs $G, H_1, H_2$ have same number of vertices, we can assume they are having same vertex set $V$ by mapping them arbitrarily.

If $|L|\geq \frac{\eps n}{20\cGabow^2}$ and $|R|\geq \frac{\eps n}{20\cGabow^2}$, then in graph $H_1$, we have \[|L|\cdot |R|\cdot d \geq \frac{\eps n}{20\cGabow^2}\cdot \frac{\eps n}{20\cGabow^2} \cdot \frac{1600\cGabow^6}{\eps^2} \geq 4\cGabow^2n^2.\] From the guarantee of \Cref{lem:ramanujan} we have $E_{H_1}(L, R)\neq \emptyset$ hence, we find $S$ computing $(s,t)$-separator across that edge.

Otherwise, we can assume WLOG $|L| < \frac{\eps n}{20\cGabow^2}$. We can also assume that $\delta \geq (1+\eps)\kappa_G$ otherwise, we would return the neighborhood of the minimum degree vertex in \Cref{alg:approx:1}. Hence from \Cref{lem:epskbalanced} we have $|L|\geq \eps \kappa_G$.

From the guarantee of \Cref{lem:vertexExpander} we have, $|N_{H_2}(L)|\geq \frac{2}{\eps}|L| \geq 2\kappa_G$. Hence, in $H_2$ there exists a neighbour of $L$ in $R$ as $|S| = \kappa_G$, implies $E_{H_2}(L, R)\neq \emptyset$. Applying $(s,t)$-max-flow on two endpoints of an edge in $E_{H_2}(L,R)$ gives us the minimum separator.

We apply $(s,t)$-max-flow on $G$ for every edge $(s,t)\in E_{H_1}\cup E_{H_2}$ which are at most $O(n/\eps^2)$ in total. The total additional time to construct both $H_1, H_2$ is at most $O(n/\eps^2)$ from the time guarantees of \Cref{lem:ramanujan,lem:vertexExpander}.
\end{proof}

The rest of the section discusses constructing the graphs $H_1, H_2$ stated in \Cref{alg:approx}.

\paragraph{Induced sub-graph of Ramanujan expander.} Here we prove \Cref{lem:ramanujan}. We construct the first graph from Ramanujan expanders. Ramanujan expanders are graphs with the strongest spectral expansion guarantee, which by the \emph{expander mixing lemma} (see e.g.~\cite{Salil2012}), implies that any two large enough sub-sets have an edge between them. This is formally described in the following theorem taken from \cite{Gabow06}.

\begin{lemma}[Discussion above Lemma 2.4 in \cite{Gabow06}]\label{lem:gabow00}
 Given integers $n,d\in \mathbb{N}$, there exists a $d'$-regular Ramanujan expander $X=(V_X,E_X)$ with $|V_X|\leq \cGabow n$, $d\leq d'\leq 4d$ for some universal constant $\cGabow$ and can be constructed in time $O(nd)$. We also have guarantee that for any $L,R\subseteq V_X$, if $|L|\cdot |R|\cdot d' \geq 4|V_X|^2$, then $E_X(L,R)\neq \emptyset$.
\end{lemma}

\begin{proof}[Proof of \Cref{lem:ramanujan}]
 Given integers $n,d$ apply \Cref{lem:gabow00} and construct $d'$-regular Ramanujan expander $X=(V_X,E_X)$ in time $O(nd)$. We construct $H=(V_H,E_H)$ from $X$ as follows. Let $V_H$ be any arbitrary subset of $V_X$ of size $n$ and $E_H=E_X(V_H, V_H)$ which is the subset of edges $E_X$ whose both endpoints are in $V_H$. Hence it takes at most $O(nd)$ total time to construct the graph $H$. Maximum degree of $H$ is at most $d'\leq 4d$ from \Cref{lem:gabow00}.
 
 For any $L,R\subseteq V_H \subseteq V_X$, if we have \[|L|\cdot |R| \cdot d' \geq |L|\cdot |R| \cdot d \geq 4\cGabow^2n^2 \geq 4|V_X|^2\] as $d'\geq d$, $|L||R|d\geq 4\cGabow^2n^2$ and $|V_X| \leq \cGabow n$ respectively. From \Cref{lem:gabow00} we have guarantee that $E_X(L,R) \neq \emptyset$, hence $E_H(L,R)\neq \emptyset$.
\end{proof}

\paragraph{Small set vertex expander.} Here we prove \Cref{lem:vertexExpander}. As explained below, we create the second graph using Ramanujan expanders, which have a high spectral expansion that implies good vertex expansion.

Let $A$ be the adjacency matrix of the graph $G$ and $D$ be the diagonal matrix with $i^{th}$ diagnoral entry as $\deg(v_i)$ where $v_i$ is the $i^{th}$ vertex of graph $G$. Let $\lambda_1(G), \lambda_2(G),\dots $ be the eigenvalues of the matrix $AD^{-1}$ sorted from high to low according to their absolute values. %
We recall the following fact of Ramanujan graphs.

\begin{fact}[\cite{Lubotzky86}]\label{lem:RamanujanEigen}
Let $X=(V_X,E_X)$ be the $d'$-regular Ramanujan graph with $n'$ vertices constructed using \Cref{lem:gabow00}. We have $\lambda_2(X) \leq 2\sqrt{d'-1}/d'$.
\end{fact}

The \emph{spectral expansion} $\gamma(G)$ of the graph $G$ is defined as $1-\lambda_2(G)$.
It is known that  strong  spectral expansion implies  strong  vertex expanders as stated below.

\begin{lemma}[Lemma 4.6 from \cite{Salil2012} ``spectral expansion implies vertex expansion'']\label{lem:spectralExpansion}
 If $G$ is a regular digraph with spectral expansion $\gamma(G) = 1-\lambda_2(G)$ then, for every $\alpha\in [0,1]$, $G$ is an $(\alpha, \frac{1}{\alpha + (1-\alpha)\lambda_{2}(G)^2}-1)$\footnote{the lemma has extra $-1$ in the expansion parameter resulting due to change in the definition of $(\alpha,\beta)$-vertex expander.}-vertex expander.
\end{lemma}

Note that the statement works for undirected graphs by considering any undirected graph as a digraph by having two directed edges in both directions for each undirected edge. 
Now, we conclude the following claim.

\begin{claim}\label{clm:bigVertexExpander}
Given an integer $n$ and approximation parameter $\eps$, we can construct a $(\frac{\eps}{10\cGabow},\frac{4}{\eps})$-vertex expander, $X=(V_X,E_X)$ of size at most $\cGabow n$ and maximum degree $O(1/\eps)$ in time $O(n/\eps)$.
\end{claim}

\begin{proof}
Let $X$ be the Ramanujan graph constructed with $n, d=40\cGabow/\eps$ as inputs. From the guarantees of \Cref{lem:gabow00} we know that $|V_X|\leq \cGabow n$ and is $d'$-regular with $d \leq d' \leq 4d$. From \Cref{lem:RamanujanEigen} we know that the second eigenvalue of $X$, $\lambda_2(X) \leq 2\sqrt{d'-1}/d' \leq 2/\sqrt{d}$.

Hence, by setting $\alpha = \eps/10\cGabow$ in \Cref{lem:spectralExpansion} we can conclude that $X$ is a $(\eps/10\cGabow, \beta)$ where $\beta = \frac{1}{(1-\alpha)\lambda_2(X)^2 + \alpha}-1 \geq \frac{1}{\frac{4}{d}+\frac{\eps}{10\cGabow}}-1 \geq \frac{4}{\eps}$ as $\cGabow\geq 1$ and $\eps \leq 1$. The time to construct the graph $X$ is at most $O(nd) = O(n/\eps)$.
\end{proof}

Note that the graph $X$ we constructed above is of size at most $\cGabow n$, but we need a vertex expander of size $n$. Taking any induced sub-graph of size $n$ would not solve the problem as it might not be a vertex expander. Below, we show that contracting groups of vertices where each group is small preserves the vertex expansion. 

\begin{claim}\label{clm:expanderSizing}
If $G = (V_G, E_G)$ is a $(\alpha,\beta)$-vertex expander with $|V_G|=\rho n, |E_G|=m$ for some parameter $\rho>1$, then we can construct a $(\frac{\alpha}{\ceil{\rho}},\frac{\beta\floor{\rho}}{\ceil{\rho}})$-vertex expander $H=(V_H, E_H)$ of size $n$ in $O(m)$ time.
\end{claim}

\begin{proof}
If $\rho$ is an integer, we obtain $H$ by contracting arbitrary groups of $\rho$ vertices in $G$. However, if $\rho$ is not an integer, we contract arbitrary groups of $\floor{\rho}$ or $\ceil{\rho}$ many vertices together to achieve a graph of size exactly $n$.

Let $L_H\subset V_H$ is of size at most $\alpha n/\ceil{\rho}$. For each $v\in V_H$ let $C_v \subset V_G$ denote the set of vertices in $G$ contracted to form $v$. Uncontracting the set of vertices in $L_H$, we get $L_G = \cup_{v\in L_H} C_v$. We have $|L_G| \leq \ceil{\rho}|L_H| \leq \alpha n$, hence we have $|N_G(L_G)| \geq \beta |L_G|$ as $G$ is a $(\alpha,\beta)$-vertex expander.

Let $C[N_G(L_G)] = \{v \in V_H \mid C_v \cap N_G(L_G) \neq \emptyset\}$ be the set of contracted vertices in $V_H$ such that corresponding set $C_v\subset V_G$ has at least one vertex from $N_G(L_G)$. Observe that $C[N_G(L_G)] = N_H(L_H)$. Every vertex in $C[N_G(L_G)]$ has at least one edge from a vertex in $L_H$ as the edge across $L_G$ and $N_G(L_G)$ is preserved after contraction. Hence, $C[N_G(L_G)]\subseteq N_H(L_H)$. Every vertex $v \in N_H(L_H)$ has a vertex $u \in L_H$ such that $(u,v)\in E_H$. Uncontracting $u,v$ we have every vertex $w\in C_v$ that has edges from some vertex $x\in C_u\subset L_G$ should belong to $N_G(L_G)$, hence $v\in C[N_G(L_G)]$ and we have $N_H(L_H)\subseteq C[N_G(L_G)]$.

We have \[|N_H(L_H)| = |C[N_G(L_G)]|\geq \frac{|N_G(L_G)|}{\ceil{\rho}} \geq \frac{\beta |L_G|}{\ceil{\rho}} \geq \frac{\beta\floor{\rho}}{\ceil{\rho}}|L|.\] The first and third inequalities follow as we contract at most $\ceil{\rho}$ and at least $\floor{\rho}$ many vertices in $G$ to construct $H$. Hence, expansion is at least $\frac{\beta\floor{\rho}}{\ceil{\rho}}$.

The time taken to contract arbitrary vertices is at most $O(m)$ by just going over each edge and mapping them to the new contracted vertices.
\end{proof}

We now conclude the proof of \Cref{lem:vertexExpander}.

\begin{proof}[Proof of \Cref{lem:vertexExpander}]
 With $n,\eps$ as inputs to \Cref{clm:bigVertexExpander}, we construct a $(\frac{\eps}{10\cGabow}, \frac{4}{\eps})$-vertex expander of size $\rho n$ where $1\leq \rho \leq \cGabow$ in time $O(n/\eps)$ with $O(n/\eps)$ edges. We now use \Cref{clm:expanderSizing} to reduce the size of the graph and construct a $(\frac{\eps}{10\cGabow\ceil{\rho}},\frac{4\floor{\rho}}{\eps \ceil{\rho}})$-vertex expander in time $O(n/\eps)$, which is also a $(\frac{\eps}{20\cGabow^2}, \frac{2}{\eps})$-vertex expander as $\cGabow \geq 1$. Hence, the total time is at most $O(n/\eps)$.
\end{proof}
 

%% file: open.tex
\section{Open Problems}

\paragraph{Exact vertex connectivity.}
It remains elusive whether there exists a linear-time \emph{deterministic} algorithm for vertex connectivity as postulated by \cite{AhoHU74}.
Even an algorithm with running time $\Ohat(mn)$, independent from $k$, is already very interesting as asked by Gabow \cite{Gabow06}. 
The currently fastest deterministic algorithm can be obtained by plugging an almost-linear-time max flow algorithm by \cite{BrandCKLPGSS23} into the algorithm by Gabow \cite{Gabow06}. This will result in an algorithm that takes $\Ohat(m(n+k\min\{k,\sqrt{n}\}) = \Ohat(mn^{1.5})$ time.

\paragraph{Approximate vertex connectivity.}
Can we obtain an $(1+\eps)$-approximation algorithm using  $o(n)$ max-flow calls? The problem seems challenging to us even when we promise that there exists a vertex mincut $(L,S,R)$ where $|L|,|S|,|R| = \Omega(n)$. Note that if randomization is allowed in this case, we can solve the problem with high probability by sampling $O(\log n)$ vertex pairs $(s,t)$ and compute $(s,t)$ vertex mincut, taking a total time of only $O(\log n)$ max flow calls.

\paragraph{Steiner vertex connectivity.}
Our terminal reduction algorithm (\Cref{thm:main terminal reduction}) does not guarantee that $T' \subset T$, i.e., the output terminal is a subset of the input terminal. Is it possible to strengthen our algorithm to guarantee this with the same running time? This would imply a \emph{deterministic} $\Ohat(mk^2)$-time algorithm for the Steiner vertex mincut problem, matching the conditional lower bound for dense graphs~\cite{HuangLSW23}. A randomized algorithm for Steiner vertex mincut with near-optimal $\Ohat(mk^2)$ time is implicit in the literature.\footnote{To see this, let $T$ be the terminal set and $(L,S,R)$ be a Steiner vertex mincut of size less than $k$.
The algorithm starts by sampling each terminal with probability $1/k$ and so, with probability $\Omega(1/k^2)$, we have $T\cap S$ is empty, but both $T \cap L$ and $T \cap R$ are not empty. Assuming this event, one can use the isolating vertex cut technique together with further subsampling (see \cite{li_vertex_2021}) to detect a Steiner mincut using $\log^{O(1)} n$ max-flows. By repeating the algorithm $O(k^2 \log n)$ times, we will detect a Steiner mincut with high probability.}